\documentclass[submission, Phys]{SciPost}
\usepackage{verbatim}
\usepackage{hyperref}
\def\doi{http://dx.doi.org/}
\usepackage{amsmath,braket,mathdots,mathtools,amssymb,xcolor,calligra,color}
\usepackage{amsthm} % JJ
\usepackage{bbm}  
\usepackage{fancyhdr}
\usepackage{float}
\usepackage{bm}
\usepackage{authblk}
\usepackage{color}
\usepackage{authblk}
\usepackage{cite} % JJ
\usepackage{graphicx}
\usepackage{tikz}
\usetikzlibrary{spy}
\usetikzlibrary{patterns}
\usetikzlibrary{calc,arrows,positioning}
\usetikzlibrary{arrows,shadows}
\usetikzlibrary{decorations.pathmorphing}	% For Feynman 
\usetikzlibrary{decorations.markings}
\usepackage{pgfplots}

\usepackage{pgfplotstable}
\usepackage{filecontents}
\usepackage[utf8]{inputenc}    

\newcommand{\be}{\begin{equation}}
\newcommand{\ee}{\end{equation}}
\newcommand{\bea}{\begin{eqnarray}}
\newcommand{\eea}{\end{eqnarray}}

\def\XXint#1#2#3{{\setbox0=\hbox{$#1{#2#3}{\int}$}
     \vcenter{\hbox{$#2#3$}}\kern-.5\wd0}}

\let\OLDthebibliography\thebibliography
\renewcommand\thebibliography[1]{
  \OLDthebibliography{#1}
  \setlength{\parskip}{0pt}
  \setlength{\itemsep}{0pt plus 0.3ex}
}

\newcommand{\sign}{\,{\rm sgn}\,}

\newcommand{\1}{\,\pmb{1}}
\newcommand{\D}[1]{{\rm d}#1}

\newtheorem{theorem}{Theorem}
\newtheorem{conjecture}{Conjecture}

\newtheorem{property}{Lemma}

\hypersetup{
     colorlinks=true,
     linkcolor=blue,
     filecolor=blue,
     citecolor = green,      
     urlcolor=cyan,
     }

\begin{document}
%%%%%%%%%%%%%%%%%%%%%%%%%%%%%%%%%%%%%%%
\begin{center}
{\Large\bf Regularization of a strong-weak duality between pointlike interactions in one dimension
%Fermionic formulation of the Lieb-Liniger model
}
\end{center}
%%%%%%%%%%%%%%%%%%%%%%%%%%%%%%%%%%%%%%%
\begin{center}
Etienne Granet\textsuperscript{1}
\end{center}
\begin{center}
{\bf 1} Kadanoff Center for Theoretical Physics, University of Chicago, 5640 South Ellis Ave, Chicago, IL 60637, USA
\end{center}
\date{\today}

\section*{Abstract}
{\bf 
%In 1D quantum physics, a pointlike interaction is equivalent to imposing boundary conditions on the wave function and its derivative at the position of the potential. 
Pointlike interactions between bosons in 1D are related to pointlike interactions between fermions through the Girardeau mapping. This mapping is a strong-weak duality since the coupling constants in the bosonic and fermionic cases are inversely proportional to each other. 
%The physical meaning of these pointlike interactions is to approximate very-short range potentials, hence have to be understood as the zero-range limit of some smooth potentials. 
We present a regularization of these pointlike interactions that preserves the strong-weak duality, contrary to previously known regularizations. This is proven rigorously. This allows one to use this duality perturbatively and we illustrate it in the Lieb-Liniger model at strong coupling.

%
%It is well-known that the only pointlike interaction between bosons in 1D quantum mechanics induces a discontinuity in the derivative of the wave function $\psi'(0^+)-\psi'(0^-)=2c\psi(0)$, with $c$ the coupling constant. Conversely, the only pointlike interaction between fermions induces a discontinuity in the wave function itself $\psi(0^+)-\psi(0^-)=2\beta\psi'(0)$, with $\beta$ the coupling constant. These two pointlike interactions are mapped to each other through the Girardeau mapping with $\beta=1/c$ We present a regularization of this pointlike interaction in terms of a smooth potential that is \emph{weakly coupled} at small $\beta$. This is in contrast with the previously known regularizations that are strongly coupled, and allows one to treat the discontinuity perturbatively around free fermions. We present an application of this potential to a strong coupling expansion of the Lieb-Liniger model.
}

%\tableofcontents
\section{Introduction}
It is well-known that bosons interacting in 1D quantum physics via a pointlike potential have a wave-function that is continuous when two bosons coincide, but with a discontinuous derivative \cite{Griffiths}. This $\delta$ interaction approximates a very short-range regular potential, with a coupling strength given by the integral of the potential. If the potential is weak then the cusp in the wave function is small, and the potential can be treated perturbatively around free bosons.

Conversely, fermions interacting via a pointlike potential have a wave-function that is \textit{discontinuous} when fermions coincide, and with a continuous derivative \cite{seba1,Albeverio}. This is dual to the bosonic case, in the sense that pointlike-interacting fermions can be exactly mapped to pointlike-interacting bosons via the Girardeau mapping \cite{G60,CS99,YG05}. The amplitudes of the discontinuity in both cases are inversely proportional, and for this reason the Girardeau mapping can also be seen as a duality between strong and weak interactions. This mapping is particularly useful in the context of quantum gases, providing an exact correspondance between the Lieb-Liniger model and the Cheon-Shigehara model \cite{LiebLiniger63a,CS99}, namely between strongly coupled bosons and weakly coupled fermions \cite{kinoshita}.

But in contrast with bosons, the interpretation of this pointlike potential between fermions in terms of distributions or as the zero-range limit of smooth potentials is problematic and has attracted attention in the past \cite{seba2,ADK98,kurasov,PB98,S99,S03,GO04,BC05,DSGDDK09,carreau,chernoff,CS98,albeverio00approximation,coutinho,RT96}. Regularizations of related pointlike interactions is still a topic of recent research \cite{christiansen03on,albeverio13a,calcada14distributional,fassari09on,albeverio15the,zolotaryuk19point,zolotaryuk17families,zolotaryuk15an,gadella,gadella2,kostenko}. At this day, the Cheon-Shigehara potential \cite{CS98} provides the most physical and operational interpretation of this interaction, since it is built as the zero-range limit of a Hermitian potential. However it breaks the strong-weak duality, since the regularized potential between fermions is also \textit{strongly coupled} for strongly coupled bosons. For this reason it is not amenable to perturbative expansion around free fermions. Up to now, no regularization of this fermionic pointlike interaction in terms of a smooth, regular Hamiltonian preserves the strong-weak duality.

%
%Interpretations in terms of generalizations of the notion of distribution have been proposed \cite{}, but lack the physical picture of an approximation to a very short range potential. Moreover, it is unclear how to translate them in second quantization \cite{}, as would be necessary for the applications to quantum gases. There are known regular potentials that approximate this jump in the fermionic wave function \cite{}, most notably the Cheon-Shigehara potential \cite{}
%
% is problematic \cite{seba2,coutinho,carreau,chernoff,CS98,ADK98,kurasov,PB98}.  ; 
 
% however they are still \textit{strongly coupled} when the discontinuity is small, contradicting the strong-weak duality. For this reason no perturbative expansion around free fermions is possible, and in all the known regularizations the discontinuity of the wave function is a non-perturbative effect \cite{coutinho,carreau,RT96,chernoff,CS98}. 

The objective of this paper is to complete the companion paper \cite{GBE21} and to introduce a regularization of this pointlike interaction between fermions that preserves the strong-weak duality. Namely, as in \cite{GBE21}, we provide a smooth potential that in the zero-range limit produces a discontinuity in the fermionic wave function, that is proportional to the strength of the potential. This fact is proven rigorously. 
%The new important feature of these potentials is that they are \textit{weakly coupled} for small discontinuities, meaning that they can be considered as a perturbation around free fermions. 
% provide numerical and analytical evidence that perturbation theory and the regularization limit commute, namely that one can perform a perturbative expansion in the coupling constant at fixed regulator, and send the regulator to zero in the truncated expansion afterwards.
In order to show that our potential allows for a perturbative treatment, we present an application of our findings to the Lieb-Liniger model with coupling $c$ \cite{LiebLiniger63a,Brezin64,korepin}, dual to the Cheon-Shigehara model with coupling $1/c$. Our potential allows for a well-behaved perturbative expansion of the energy levels of the Cheon-Shigehara model at small coupling, and we show that we recover indeed the energy levels of the Lieb-Liniger model at strong coupling. This duality should prove useful to study expansions of the Lieb-Liniger model at large coupling \cite{jimbo,creamer,korepincorr,brandcherny,deNP15,GE20,GE21}.

\section{Pointlike interactions as self-adjoint extensions\label{result}}
\subsection{Self-adjoint extensions}
In 1D quantum physics, a Hamiltonian describing a pointlike interaction at position $x$ can be modelled by a self-adjoint extension of a free Hamiltonian on a region containing a neighbourhood of $x$. Namely, one considers a Hamiltonian whose matrix elements between wave functions $\psi,\phi$ are
\begin{equation}
\langle \phi|H|\psi\rangle=-\int_{-L/2}^{L/2}\phi^*(x)\left(\frac{\D{}}{\D{x}} \right)^2\psi(x)\D{x}\,,
\end{equation}
with e.g. periodic boundary conditions on a ring of size $L$, and one looks for the possible Hilbert spaces of functions that are smooth and integrable on $-L/2<x<0$ and $0<x<L/2$, such that $H$ is hermitian. This corresponds to imposing particular connection conditions on the wave function at the position of the potential $x=0$. The most general connection conditions on the wave function $\psi_\pm$ and its derivative $\psi'_\pm$ before and after the potential are \cite{Albeverio,seba1} 
\begin{equation}\label{gen}
\begin{cases}
\text{either}\quad\left(\begin{matrix}
\psi_+\\\psi_+'
\end{matrix} \right)=\Lambda\left(\begin{matrix}
\psi_-\\\psi_-'
\end{matrix} \right)\,,\qquad \Lambda=e^{i\theta}\left(\begin{matrix}
a&b\\c&d
\end{matrix} \right)\\
\text{or}\quad \psi'_\pm=h_\pm \psi_\pm
\end{cases}
\end{equation}
with $a,b,c,d,\theta$ reals that satisfy the constraint $ad-bc=1$, and $h_\pm$ reals. We will call $\Lambda$ the interaction matrix. Since the second possibility $\psi'_\pm=h_\pm \psi_\pm$ acts like a hard wall by decoupling $x>0$ and $x<0$,  we will focus only on the first possibility.

\subsection{Pointlike interactions between bosons}
From the general interaction matrix \eqref{gen}, one can investigate the possible pointlike interactions between bosons. This constrains the wave function to be even, so that $\psi_+=\psi_-$ and $\psi'_+=-\psi'_-$. This constrains $e^{i\theta}=1$ and $a=d$, leaving a two-parameter family of interactions whose interaction matrix is given by
\begin{equation}\label{boson}
\Lambda=\left(\begin{matrix}
a&b\\c&a
\end{matrix} \right)\,.
\end{equation}
For these interactions the bosonic wave function is continuous, but has a discontinuity in the derivative equal to $2\gamma \psi_+$ with $\gamma=\tfrac{a-1}{b}\in[-\infty,\infty]$. One notes that fixing $\gamma$ still leaves a one-parameter family of interactions. Varying this remaining parameter has no effect on bosons, but affects the behaviour of fermions interacting via the same potential. One particular value of this remaining parameter makes the potential transparent for fermions and corresponds to the interaction matrices
\begin{equation}\label{pureboson}
\Lambda=\left(\begin{matrix}
1&0\\2\gamma&1
\end{matrix} \right)\,.
\end{equation}
This purely bosonic pointlike interaction is the textbook "$\delta$-interaction", that owns its name to the fact that this discontinuity of the derivative can be produced by the following Schr\"odinger equation
\begin{equation}\label{eqdelta}
\psi''(x)+k^2\psi(x)=2\gamma \delta(x)\psi(x)\,.
\end{equation}
Indeed, by integrating on a small window $-\epsilon<x<\epsilon$ around $0$ and applying the usual rules on the $\delta$ distribution we obtain the condition $\psi'(0^+)-\psi'(0^-)=2\gamma \psi(0)$. However this manipulation should be considered as heuristic, since the product of a distribution with a non-smooth function is ill-defined. The proper way of understanding \eqref{eqdelta} and the boundary condition it induces is thus through \emph{regularization} of the $\delta$ distribution. It means that if one replaces the distribution $\delta(x)$ by a smooth function $\delta_a(x)\equiv \frac{1}{a}\Delta(x/a)$ with 
\begin{equation}
\int \Delta(x)\D{x}=1\,,
\end{equation}
then the even wave function $\psi_a(x)$ solution to this regularized equation 
\begin{equation}\label{shrod1}
\psi_a''(x)+k^2\psi_a(x)=2\gamma \delta_a(x)\psi_a(x)\,,
\end{equation}
satisfies when $a\to 0$
\begin{equation}\label{shrod0}
\begin{cases}
 \psi''(x)+k^2\psi(x)=0\,,\qquad \text{for }x\neq 0\\
\psi(0+)=\psi(0-)\\
\psi'(0+)-\psi'(0-)=2\gamma\psi(0)
\end{cases}\,.
\end{equation}
 This interpretation is also physical: a pointlike potential is nothing but an idealized version of a very short-range smooth potential $\delta_a(x)$ convenient to the theoretician.

\subsection{Pointlike interactions between fermions}
One can as well investigate the possible pointlike interactions for fermions, i.e. for which the resulting wave function can be odd. One checks that this constraint leads to the same general interaction matrix \eqref{boson} as for bosons. This time, the derivative of the fermionic wave function is continuous, but the wave-function itself acquires a discontinuity equal to $2\beta \psi_\pm$ with $\beta=\tfrac{b}{1+a}$. Fixing $\beta$ still leaves a one-parameter family of interactions. Varying this remaining parameter has no effect on fermions, but changes the behaviour of bosons that interact via the same potential. One particular value of this remaining parameter makes the potential transparent for bosons and corresponds to the interaction matrices
\begin{equation}\label{free}
\Lambda=\left(\begin{matrix}
1&2\beta\\0&1
\end{matrix} \right)\,.
\end{equation}
Another choice is an interaction that is hard-core for bosons, namely that impose a bosonic wave function to vanish at the position of the potential. This corresponds to the interaction matrix
\begin{equation}\label{core}
\Lambda=\left(\begin{matrix}
-1&0\\-2\beta&-1
\end{matrix} \right)\,.
\end{equation}
These two interaction matrices have the same effect on fermions, and implement the connection conditions
\begin{equation}\label{shrod78}
\begin{cases}
\psi'(0^+)=\psi'(0^-)\\
\psi(0^+)-\psi(0^-)=2\beta\psi'(0)
\end{cases}\,,
\end{equation}
for odd wave functions $\psi$. But in contrast with the bosonic pointlike interaction, the realization of this discontinuity through a regular or distributional potential is problematic and has attracted a lot of controversies in the past. In particular, its widespread name "$\delta'$-interaction" is improper \cite{seba2}. In fact, the heuristic analog of the Schr\"odinger equation \eqref{eqdelta} to this case is \cite{ADK98}
\begin{equation}\label{eqdelta2}
\psi''(x)+k^2\psi(x)=2\beta \partial_x[\delta(x)\partial_x]\psi(x)\,.
\end{equation}
Indeed, by integrating $x$ for $-\epsilon<x<y$ and then $y$ for $-\epsilon<y<\epsilon$ and applying the usual rules on the $\delta$ distribution yields $\psi(0^+)-\psi(0^-)=2\beta \psi'(0)$. But we stress that this manipulation is again heuristic, since the product of distributions $\partial_x[\delta(x)\partial_x]\psi(x)$ is  ill-defined, and requires a generalization of distributions to discontinuous test functions \cite{kurasov,PB98}. Although mathematically sound, these generalizations suffer from a lack of physical meaning since one loses the interpretation of the potential as an approximation of a very short range smooth potential. Contrary to the previous case \eqref{eqdelta} indeed, regularizing the distribution $\delta(x)$ into a smooth potential $\delta_a(x)$ in \eqref{eqdelta2} \textit{does not} yield the expected discontinuity when $a\to 0$. Namely, the odd solution to
\begin{equation}\label{eqdelta3}
\psi_a''(x)+k^2\psi_a(x)=2\beta \partial_x[\delta_a(x)\partial_x]\psi_a(x)\,.
\end{equation}
is \textit{continuous} at $0$ when $a\to 0$, as shown in Appendix \ref{app}. 

It is thus a non-trivial problem to find a regular potential $V_a(x)$ such that the odd solution $\psi_a(x)$ to the Schr\"odinger equation
\begin{equation}\label{shord4}
\psi_a''(x)+k^2\psi_a(x)=V_a(x)\psi_a(x)\,,
\end{equation}
satisfies when $a\to 0$ 
\begin{equation}\label{shrod5}
\begin{cases}
 \psi''(x)+k^2\psi(x)=0\,,\qquad \text{for }x\neq 0\\
\psi'(0^+)=\psi'(0^-)\\
\psi(0^+)-\psi(0^-)=2\beta\psi'(0)
\end{cases}\,.
\end{equation}
Different solutions to this problem in terms of non-Hermitian potentials, non-local potentials or pseudo-potentials were proposed \cite{carreau,chernoff,RT96}. The first solution in terms of a Hermitian, regular potential was found by Cheon and Shigehara \cite{CS98} and reads
\begin{equation}\label{cheon}
V_{a,\beta}(x)=\left( \frac{1}{\beta}-\frac{1}{a}\right)\delta(x+a)+2\left(\frac{\beta}{a^2}-\frac{1}{a} \right)\delta(x)+\left(\frac{1}{\beta}-\frac{1}{a} \right)\delta(x-a)\,.
\end{equation}
The $\delta$ distributions therein can be further regularized on a  smaller enough scale than $a$ \cite{exner01potential,zolotaryuk17families}. This regularization corresponds to the interaction matrix \eqref{free} that is transparent for bosons.

\section{Regularization preserving the strong-weak duality}
\subsection{The potential}
In this Section we consider the case $\beta>0$ of the interaction matrix \eqref{core} and the induced connection conditions \eqref{shrod78}. In this case we show that it can be regularized into a smooth potential that is small for small $\beta$. Namely, we prove the following Theorem.
\begin{theorem}\label{theorem}
We consider $\beta>0$ and define the following potential
\begin{equation}\label{potential}
V_{a,\beta}(x)= \frac{\beta\sigma''_a(x)}{x+\beta \sigma_a(x)}\,,
\end{equation}
with $\sigma_a(x)\equiv \sigma(x/a)$ where $\sigma(x)$ is any odd regular function that satisfies
\begin{equation}\label{sigma}
\begin{aligned}
\underset{x\to\infty}{\lim}\, \sigma(x)&=1\\
\forall x,\,\sigma'(x)&\geq 0\\
\sigma'(0)&> 0\\
\underset{x\to\infty}{\lim}\, x^2\sigma''(x)&=0\,.
\end{aligned}
\end{equation}
(For example, one can choose $\sigma_a(x)=\tanh \tfrac{x}{a}$.)
Denote $\psi_{a,\beta}(x)$ the odd solution to the Schr\"odinger equation
\begin{equation}\label{eqv}
-\psi_{a,\beta}''(x)+V_{a,\beta}(x)\psi_{a,\beta}(x)=k^2\psi_{a,\beta}(x)\,,
\end{equation}
with a fixed boundary condition $\psi_{a,\beta}(1)=1$. In the limit $a\to 0$ at fixed $\beta>0$, it satisfies
\begin{equation}\label{shrod2}
\begin{cases}
\psi''(x)+k^2\psi(x)=0\,,\qquad \text{for }x\neq 0\\
\psi'(0^+)=\psi'(0^-)\\
\psi(0^+)-\psi(0^-)=2\beta\psi'(0)\,.
\end{cases}
\end{equation}
%As for the even solution to \eqref{eqv}, in the limit $a\to 0$ it is continuous for all $x$, vanishes at $x=0$ and satisfies $\psi''(x)+k^2\psi(x)=0$ for $x\neq 0$.
\end{theorem}

\subsection{Proof of Theorem \ref{theorem} \label{nonpert}}
\subsubsection{Case $k=0$}
We consider first $k=0$. The Schr\"odinger equation \eqref{eqv} is a linear differential equation of degree $2$ with an even potential, so it has two linearly independent solutions, one odd $\psi^0_{a,\beta}$ and one even $\phi^0_{a,\beta}$. One verifies that the following two wave functions satisfy  \eqref{eqv}  for $k=0$ indeed
\begin{equation}
\begin{aligned}
\psi^0_{a,\beta}(x)&=x+\beta\sigma_a(x)\\
\phi^0_{a,\beta}(x)&=\frac{1}{1+\beta\sigma_a'(x)}+(x+\beta \sigma_a(x))\int_{0}^x \frac{\beta\sigma_a''(y)}{(y+\beta\sigma_a(y))(1+\beta\sigma_a'(y))^2}\D{y}\,.
\end{aligned}
\end{equation}
The odd solution $\psi^0_{a,\beta}(x)$ has a well-defined limit when $a\to 0$
\begin{equation}
\psi^0_{0,\beta}(x)=x+\beta\sign(x)\,,
\end{equation}
that satisfies \eqref{shrod2}.

\subsubsection{A self-consistent equation for $k\neq 0$}
Using the method of variation of parameters, the solutions to the equation
\begin{equation}
\psi''(x)-V_{a,\beta}(x)\psi(x)=f(x)\,,
\end{equation}
for a given arbitrary function $f(x)$ are 
\begin{equation}
\begin{aligned}
\psi(x)=&(x+\beta\sigma_a(x))\int_{0}^x f(y)\phi^0_{a,\beta}(y)\D{y}-\phi^0_{a,\beta}(x)\int_{0}^x f(y) (y+\beta\sigma_a(y))\D{y}\\
&+A(x+\beta\sigma_a(x))+B\phi^0_{a,\beta}(x)\,,
\end{aligned}
\end{equation}
with $A,B$ integration constants. By setting $f(x)=-k^2\psi(x)$ and imposing $\psi(x)$ odd, one obtains a self-consistent equation for $\psi_{a,\beta}(x)$ for $k\neq 0$
\begin{equation}
\begin{aligned}\label{self2}
\psi_{a,\beta}(x)=&-k^2(x+\beta\sigma_a(x))\int_{0}^x \psi_{a,\beta}(y)\phi^0_{a,\beta}(y)\D{y}+k^2\phi^0_{a,\beta}(x)\int_{0}^x \psi_{a,\beta}(y) (y+\beta\sigma_a(y))\D{y}\\
&+A_a(x+\beta\sigma_a(x))\,,
\end{aligned}
\end{equation}
where we made explicit the $a$ dependence in the integration constant $A_a$, that is such that $\psi_{a,\beta}(1)=1$. Let us now fix an integration constant $A$ independent of $a$, and define $\tilde{\psi}_{a,\beta}(x)$ as the odd solution to \eqref{eqv} but with an $a$-dependent boundary condition at $x=1$ such that
\begin{equation}
\begin{aligned}\label{self}
\tilde{\psi}_{a,\beta}(x)=&-k^2(x+\beta\sigma_a(x))\int_{0}^x \tilde{\psi}_{a,\beta}(y)\phi^0_{a,\beta}(y)\D{y}+k^2\phi^0_{a,\beta}(x)\int_{0}^x \tilde{\psi}_{a,\beta}(y) (y+\beta\sigma_a(y))\D{y}\\
&+A(x+\beta\sigma_a(x))\,.
\end{aligned}
\end{equation}
We have
\begin{equation}\label{ewd}
\psi_{a,\beta}(x)=\frac{\tilde{\psi}_{a,\beta}(x)}{\tilde{\psi}_{a,\beta}(1)}\,.
\end{equation}
Given that \eqref{shrod2} is invariant under multiplication of the wave function by a constant, the problem is equivalent to showing that $\tilde{\psi}_{a,\beta}$ satisfies \eqref{shrod2} in the limit $a\to 0$, and that $\tilde{\psi}_{a,\beta}(1)$ has a finite non-zero limit when $a\to 0$.

\subsubsection{Study of $\phi^0_{a,\beta}(x)$\label{phihi}}
To go further, one needs to study the function $\phi^0_{a,\beta}(x)$ when $a\to 0$. Let us define the quantity
\begin{equation}
I_a(x)=\int_0^x \frac{\beta\sigma_a''(y)}{(y+\beta\sigma_a(y))(1+\beta\sigma_a'(y))^2}\D{y}\,.
\end{equation}
Considering $x>0$ and a fixed $C>0$, we do a change of variable $y=at$ to write
\begin{equation}\label{equt}
I_a(x)=a\int_0^C \frac{\beta\sigma''(t)}{(at+\beta\sigma(t))(a+\beta\sigma'(t))^2}\D{t}+a\int_C^{x/a} \frac{\beta\sigma''(t)}{(at+\beta\sigma(t))(a+\beta\sigma'(t))^2}\D{t}\,.
\end{equation}
Using $\sigma(t),\sigma'(t)\geq 0$, the integrand of the first term is bounded by $|\tfrac{\sigma''(t)}{\beta^2\sigma(t)\sigma'(t)^2}|$, which is integrable on $[0,C]$ since $\sigma(t)$ is odd and $\sigma'(0)>0$ by assumption. Hence the first terms goes to $0$ when $a\to 0$. We do an integration by part on the second term to write
\begin{equation}\label{equty}
\begin{aligned}
a\int_C^{x/a} \frac{\beta\sigma''(t)}{(at+\beta\sigma(t))(a+\beta\sigma'(t))^2}\D{t}=&-\frac{a}{(x+\beta\sigma(x/a))(a+\beta\sigma'(x/a))}\\
&+\frac{a}{(aC+\beta\sigma(C))(a+\beta\sigma'(C))}\\
&-a\int_C^{x/a}\frac{\D{t}}{(at+\beta\sigma(t))^2}\,.
\end{aligned}
\end{equation}
The first term on the right-hand side has a finite limit when $a\to 0$ for $x>0$ given by
\begin{equation}
\underset{a\to 0}{\lim}\, -\frac{a}{(x+\beta\sigma(x/a))(a+\beta\sigma'(x/a))}=-\frac{1}{x+\beta}\,.
\end{equation}
The second term on the right-hand side of \eqref{equty} goes to zero when $a\to 0$. To investigate the third term, we fix a $\epsilon>0$ and take $C$ large enough such that $\sigma(t)\geq 1-\epsilon$ for $t\geq C$. Then for $a\leq x/C$, one has
\begin{equation}
a\int_C^{x/a}\frac{\D{t}}{(at+\beta\sigma(t))^2}\leq a\int_C^{x/a}\frac{\D{t}}{(at+\beta(1-\epsilon))^2}=-\frac{1}{x+\beta(1-\epsilon)}+\frac{1}{aC+\beta(1-\epsilon)}\,.
\end{equation}
Moreover, since $\sigma'(t)\geq 0$ and $\sigma(t)\to 1$ when $t\to\infty$, we have $\sigma(t)\leq 1$. This yields for $a\leq x/C$
\begin{equation}
a\int_C^{x/a}\frac{\D{t}}{(at+\beta\sigma(t))^2}\geq a\int_C^{x/a}\frac{\D{t}}{(at+\beta)^2}=-\frac{1}{x+\beta}+\frac{1}{aC+\beta}\,.
\end{equation}
Since this is valid for any $\epsilon>0$, one concludes that for $x>0$
\begin{equation}
\underset{a\to 0}{\lim}\, I_a(x)=-\frac{1}{\beta}\,.
\end{equation}
From this we conclude that for $x\neq 0$ the limit $a\to 0$ of $\phi^0_{a,\beta}(x)$ is
\begin{equation}
\phi^0_{0,\beta}(x)=-\frac{|x|}{\beta}\,.
\end{equation}

Let us now show that $\phi^0_{a,\beta}(x)$ can be bounded independently of $a$ uniformly for $x\in[0,M]$, for a fixed $M>0$. Since the first term of \eqref{equt} is independent of $x$ and goes to zero when $a\to 0$, it can be bounded independently of $a$ uniformly for $x\in[0,M]$. The same holds true for the second term on the right-hand side of \eqref{equty}. The first term on the right-hand side can be bounded as
\begin{equation}
\left|\frac{a}{(x+\beta\sigma(x/a))(a+\beta\sigma'(x/a))}\right|\leq \frac{1}{x+\beta \sigma(x/a)}\,.
\end{equation}
When multiplied by $x+\beta\sigma(x/a)$ in $\phi_{a,\beta}^0(x)$, this gives a term that can be bounded uniformly as well. We now focus on the third term on the right-hand side of \eqref{equty}. Let us choose $C>0$ such that for $0<x<C$, $\sigma(x)\geq x\sigma'(0)/2$. This exists indeed since $\sigma'(0)>0$. Then, we write the bound
\begin{equation}
\frac{1}{(at+\beta\sigma(t))^2}\leq \begin{cases}
\frac{1}{(\beta\sigma(C))^2}\quad \text{for }t>C\\
\frac{1}{(a+\beta\sigma'(0)/2)^2}\frac{1}{t^2}\quad \text{for }t<C
\end{cases}\,.
\end{equation}
It follows that
\begin{equation}
\left|a\int_C^{x/a}\frac{\D{t}}{(at+\beta\sigma(t))^2} \right|\leq \begin{cases}
\frac{|x-aC|}{(\beta\sigma(C))^2}\quad \text{for }x>aC\\
\frac{4a}{(\beta\sigma'(0))^2}\left|\frac{a}{x}-\frac{1}{C} \right|\quad \text{for }x<aC
\end{cases}\,.
\end{equation}
Using that $\tfrac{\sigma(x)}{x}$ and so $\tfrac{\sigma(x/a)}{x/a}$ are bounded uniformly in $x$ independently of $a$,  we conclude thus
%
%As for the second term, with $at+\beta\sigma(t)\geq \beta\sigma(C)$ in the last integrand of \eqref{equty} we obtain
%\begin{equation}
%\left| a\int_C^{x/a} \frac{\beta\sigma''(t)}{(at+\beta\sigma(t))(a+\beta\sigma'(t))^2}\D{t}\right|\leq \frac{2}{\beta\sigma(C)}+\frac{M}{\beta^2\sigma(C)^2}\,.
%\end{equation}
%Hence $I_a(x)$ is bounded independently of $a$ uniformly in $x\in[0,M]$. Since $\sigma'(t)\geq 0$ and $\sigma_a(t)$ can be bounded independently of $a$, it follows 
that $\phi^0_{a,\beta}(x)$ can be bounded independently of $a$ uniformly in $x\in[0,M]$.

\subsubsection{Uniform bound on $\tilde{\psi}_{a,\beta}(x)$ \label{grom}}
Using that both $\sigma_a(x)$ and $\phi_{a,\beta}(x)$ can be bounded independently of $a$ uniformly in $x\in[-M,M]$, we see from \eqref{self} that one can find constants $\lambda\geq 0,\mu\geq 0$ (that depend on $M,k,\beta$, but not on $a$) such that for all $x\in[-M,M]$
\begin{equation}\label{ineq}
|\tilde{\psi}_{a,\beta}(x)|\leq \lambda \left|\int_{0}^x |\tilde{\psi}_{a,\beta}(y)|\D{y}\right|+\mu\,.
\end{equation}
We consider $x>0$ and apply Gr\"onwall's lemma. Introducing $\kappa(x)$ through
\begin{equation}
\int_{0}^x |\tilde{\psi}_{a,\beta}(y)|\D{y}=e^{\lambda x}\kappa(x)\,,
\end{equation}
one finds from \eqref{ineq}
\begin{equation}
\kappa'(x)\leq \mu e^{-\lambda x}\,.
\end{equation}
Hence
\begin{equation}
\kappa(x)\leq \mu \int_{0}^x e^{-\lambda y}\D{y}\,,
\end{equation}
and so for $x>0$
\begin{equation}
|\tilde{\psi}_{a,\beta}(x)|\leq \lambda\mu e^{\lambda x}\int_{0}^x e^{-\lambda y}\D{y}+\mu\,,
\end{equation}
which is a bound independent of $a$. Hence $\tilde{\psi}_{a,\beta}(x)$ can be bounded independently of $a$ uniformly in $x\in[-M,M]$.\\

\subsubsection{Existence of $\tilde{\psi}_{a,\beta}(x)$ in the limit $a\to 0$ \label{exi}}
We now wish to prove that the limit
\begin{equation}
\underset{a\to 0}{\lim}\, \tilde{\psi}_{a,\beta}(x)
\end{equation}
exists. 
%Let us fix $x_0>0$. The self-consistent equation \eqref{self} can also be written
%\begin{equation}
%\begin{aligned}\label{self2}
%\psi_{a,\beta}(x)=&-k^2(x+\beta\sigma_a(x))\int_{x_0}^x \psi_{a,\beta}(y)\phi^0_{a,\beta}(y)\D{y}+k^2\phi^0_{a,\beta}(x)\int_{x_0}^x \psi_{a,\beta}(y) (y+\beta\sigma_a(y))\D{y}\\
%&+B_a(x+\beta\sigma_a(x))+C_a\phi_{a,\beta}^0(x)\,,
%\end{aligned}
%\end{equation}
%with $B_a,C_a$ integration constants determined from $\psi_{a,\beta}(1)=-\psi_{a,\beta}(-1)=1$. 
Let us consider $f_{\beta}(x)$ satisfying \eqref{shrod2}. One checks that it satisfies
\begin{equation}
\begin{aligned}
f_{\beta}(x)=&-\frac{k^2}{\beta}(x+\beta\sign(x))\int_0^x f_{\beta}(y)|y|\D{y}-\frac{k^2}{\beta}|x|\int_0^x f_{\beta}(y)(y+\beta \sign(y))\D{y}\\
&+A'(x+\beta\sign(x))\,,
\end{aligned}
\end{equation}
with $A'$ to determine from the boundary condition. Let us choose it such that $A'=A$. Then from the self-consistent equation \eqref{self} we have
\begin{equation}
\begin{aligned}
\tilde{\psi}_{a,\beta}(x)-f_{\beta}(x)&=\int_{0}^x \tilde{\psi}_{a,\beta}(y)(g_a(x,y)-g_0(x,y))\D{y}+\int_{0}^x g_b(x,y)(\tilde{\psi}_{a,\beta}(y)-f_{\beta}(y))\D{y}\\
&+A\beta (\sigma_a(x)-\sign(x))\,,
\end{aligned}
\end{equation}
with
\begin{equation}
g_a(x,y)=-k^2(x+\beta\sigma_a(x))\phi^0_{a,\beta}(y)+k^2(y+\beta\sigma_a(y))\phi^0_{a,\beta}(x)\,.
\end{equation}
Since $\sigma_a(x)$, $\tilde{\psi}_{a,\beta}(x)$ and $\phi^0_{a,\beta}(x)$ can be bounded in $x$ independently of $a$, $g_a(x,y)$ can also be bounded in $x,y$ independently of $a$. Hence
\begin{equation}
\begin{aligned}
\left|\int_{0}^x g_b(x,y)(\tilde{\psi}_{a,\beta}(y)-f_{\beta}(y))\D{y}\right|\leq  \lambda \left|\int_0^x |\tilde{\psi}_{a,\beta}(y)-f_\beta(y)|\D{y}\right|\,,
\end{aligned}
\end{equation}
for some $\lambda>0$. Besides, since $\psi_{a,\beta}(y)(g_a(x,y)-g_0(x,y))$ is uniformly bounded in $y$ independently of $a$, and since it goes to zero when $a\to 0$ at fixed $y$, one can conclude from dominated convergence theorem that the quantity
\begin{equation}
 \int_{0}^x \tilde{\psi}_{a,\beta}(y)(g_a(x,y)-g_0(x,y))\D{y}
\end{equation}
is a function of $x>0$ that vanishes when $a\to 0$. The same holds for $A\beta(\sigma_a(x)-\sign(x))$. It follows that one can write
\begin{equation}
|\tilde{\psi}_{a,\beta}(x)-f_\beta(x)|\leq \lambda \left|\int_0^x |\tilde{\psi}_{a,\beta}(y)-f_\beta(y)|\D{y} \right|+\mu_a(x)\,,
\end{equation}
with $\mu_a(x)>0$ that is bounded uniformly in $x$ independently of $a$, and with $\mu_a(x)\to 0$ when $a\to 0$ for $x>0$. Repeating the arguments of Section \ref{grom}, one obtains
\begin{equation}
|\tilde{\psi}_{a,\beta}(x)-f_\beta(x)|\leq \lambda e^{\lambda x}\int_0^x \mu_a(y)e^{-\lambda y}\D{y}+\mu_a(x)\,.
\end{equation}
Since $\mu_a(x)$ is bounded uniformly in $x$, one can apply dominated convergence theorem and deduce that the limit of the left-hand side when $a\to 0$ exists and vanishes. This means
\begin{equation}
\underset{a\to 0}{\lim}\, \tilde{\psi}_{a,\beta}(x)=f_\beta(x)\,.
\end{equation}
Hence by construction it satisfies \eqref{shrod2}. Besides, since one can choose $A$ such that $f_\beta(1)\neq 0$, it follows that the limit of $\tilde{\psi}_{a,\beta}(1)$ in \eqref{ewd} exists and is non-zero. This concludes the proof of the Theorem.

\section{Commuting the limits $a\to 0$ and $\beta\to 0$  \label{pertee}}
The objective of this section is to give numerical evidence for the possibility of commuting the $a\to 0$ and $\beta\to 0$ limits in the solution to the Shr\"odinger equation \eqref{eqv}. Namely, we have the following Conjecture.
\begin{conjecture}\label{conjecture}
The function $\psi_{a,\beta}(x)$ of Theorem \ref{theorem} admits an expansion in $\beta$ at fixed $a$ and $x$
\begin{equation}\label{exp}
\psi_{a,\beta}(x)=\psi_{a}^{(0)}(x)+\beta\psi_{a}^{(1)}(x)+\beta^2\psi_{a}^{(2)}(x)+...
\end{equation}
such that the limit $a\to 0$ for $x\neq 0$ of each term $\psi_{a}^{(n)}(x)$ exists, and \eqref{shrod2} is satisfied order by order in $\beta$.
\end{conjecture}
Without loss of generality, we will focus on the solution in the segment $[0,x_0]$ for a given $x_0>0$, and assume the boundary condition
\begin{equation}\label{bc}
\psi_{a,\beta}(x_0)=1+\beta \frac{\sigma_a(x_0)}{x_0}\,,
\end{equation}
which turns out to be convenient for the calculations.  We have the following Lemma
\begin{property}\label{lemmaa}
We have for all integer $m$, at fixed $a>0$ and $x\in[0,x_0]$, with the boundary condition \eqref{bc}
\begin{equation}\label{lemma}
\psi_{a,\beta}(x)=\left(1+\beta\frac{\sigma_a(x)}{x} \right)\sum_{n=0}^m \beta^ng_n(x)+\mathcal{O}(\beta^{m+1})\,,
\end{equation}
with
\begin{equation}
g_0(x)=\frac{e^{ikx}-e^{-ikx}}{e^{ikx_0}-e^{-ikx_0}}\,,
\end{equation}
and
\begin{equation}\label{rec}
g_{n+1}(x)=\int_0^{x_0} \sigma_a(y) \left(\frac{\D{}}{\D{y}}\right)^2 \left[j(x,y)g_n(y)\right]\D{y}\,,
\end{equation}
with the function
\begin{equation}\label{j}
j(x,y)=\begin{cases}
\frac{1}{2iky}\left[e^{ik(x-y)}-e^{-ik(x-y)}-(e^{ikx}-e^{-ikx})\frac{e^{ik(x_0-y)}-e^{-ik(x_0-y)}}{e^{ikx_0}-e^{-ikx_0}}\right]\qquad \text{if }y<x\\
-\frac{1}{2iky}\left[(e^{ikx}-e^{-ikx})\frac{e^{ik(x_0-y)}-e^{-ik(x_0-y)}}{e^{ikx_0}-e^{-ikx_0}}\right]\qquad \text{if }y>x
\end{cases}\,.
\end{equation}
\end{property}
\begin{proof}
Let us prove this Lemma by recurrence on $m$. At leading order in $\beta$, the odd solution to \eqref{eqv} with the boundary condition \eqref{bc} is
\begin{equation}
\psi_{a,\beta}(x)=\frac{e^{ikx}-e^{-ikx}}{e^{ikx_0}-e^{-ikx_0}}+\mathcal{O}(\beta)\,.
\end{equation}
Let us now assume that the Lemma is true for $m$, and call $\beta^{m+1}\chi(x)$ the remainder $\mathcal{O}(\beta^{m+1})$ in \eqref{lemma}. Using the method of variation of parameters, the odd solution to the equation
\begin{equation}
\psi''(x)+k^2\psi(x)=f(x)\,,
\end{equation}
for a given arbitrary even function $f(x)$ is 
\begin{equation}
\psi(x)=\frac{1}{2ik}\int_0^ x f(y)(e^{ik(x-y)}-e^{-ik(x-y)})\D{y}+A(e^{ikx}-e^{-ikx})\,,
\end{equation}
with $A$ an integration constant. Applying this to \eqref{eqv} for $\psi=\chi$ and
\begin{equation}
f(x)=\frac{\sigma_a''(x)}{x}\sum_{k=0}^m \left(-\sigma_a(x) \right)^{m-k}\left(g_k(x)+\sigma_a(x)g_{k-1}(x)\right)\,,
\end{equation}
with by convention $g_{-1}(x)=0$, we find
\begin{equation}
\chi(x)=\frac{1}{2ik}\int_0^x \sigma_a''(y)\frac{g_m(y)}{y}(e^{ik(x-y)}-e^{-ik(x-y)})\D{y}+A(e^{ikx}-e^{-ikx})+\mathcal{O}(\beta)\,,
\end{equation}
with $A$ an integration constant. Integrating by part twice, we obtain
\begin{equation}
\begin{aligned}
\chi(x)&=\frac{1}{2ik}\int_0^x \sigma_a(y) \left(\frac{\D{}}{\D{y}}\right)^2\left[\frac{g_m(y)}{y}(e^{ik(x-y)}-e^{-ik(x-y)})\right]\D{y}\\
&+\frac{\sigma_a(x)}{x}g_m(x)-\frac{\sigma_a'(0)g_m'(0)}{2ik}(e^{ikx}-e^{-ikx})+A(e^{ikx}-e^{-ikx})+\mathcal{O}(\beta)\,.
\end{aligned}
\end{equation}
Hence we obtain
\begin{equation}
\psi_{a,\beta}(x)=\left(1+\beta\frac{\sigma_a(x)}{x} \right)\sum_{n=0}^{m+1} \beta^ng_n(x)+\mathcal{O}(\beta^{m+2})\,,
\end{equation}
with
\begin{equation}
\begin{aligned}
g_{m+1}(x)&=\frac{1}{2ik}\int_0^x \sigma_a(y) \left(\frac{\D{}}{\D{y}}\right)^2\left[\frac{g_m(y)}{y}(e^{ik(x-y)}-e^{-ik(x-y)})\right]\D{y}\\
&-\frac{\sigma_a'(0)g_m'(0)}{2ik}(e^{ikx}-e^{-ikx})+A(e^{ikx}-e^{-ikx})\,.
\end{aligned}
\end{equation}
Now, imposing the boundary condition \eqref{bc} to determine $A$, and using that $g_n(x_0)=0$ for $n>0$, yields the form \eqref{rec} for $g_{m+1}(x)$.
\end{proof}

%\subsection{Numerical evidence for Conjecture \ref{conjecture} \label{conjectureed}}
To investigate Conjecture \ref{conjecture} numerically, we use Lemma \ref{lemmaa} to write for any $m$
\begin{equation}
\psi_{a,\beta}(x)=\sum_{n= 0}^m\psi_{(n)}(x)\beta^n+\mathcal{O}(\beta^{m+1})\,,
\end{equation}
with
\begin{equation}
\psi_{(n)}(x)=g_n(x)+\frac{\sigma_a(x)}{x}g_{n-1}(x)\,,
\end{equation}
with $g_n(x)$ satisfying the recurrence \eqref{rec} and $g_{-1}(x)=0$. This can be evaluated numerically efficiently. To check the validity of Conjecture \ref{conjecture}, one has to ensure that $\psi_{(n)}(x)$ has a finite limit when $a\to 0$, and that in this limit one has $\psi_{(n)}(0^+)=\psi'_{(n-1)}(0^+)$. This is shown to be observed for the first four orders in $\beta$ in Figure \ref{ctdensity4}.

\begin{figure}[H]
\begin{center}
 \includegraphics[scale=0.35]{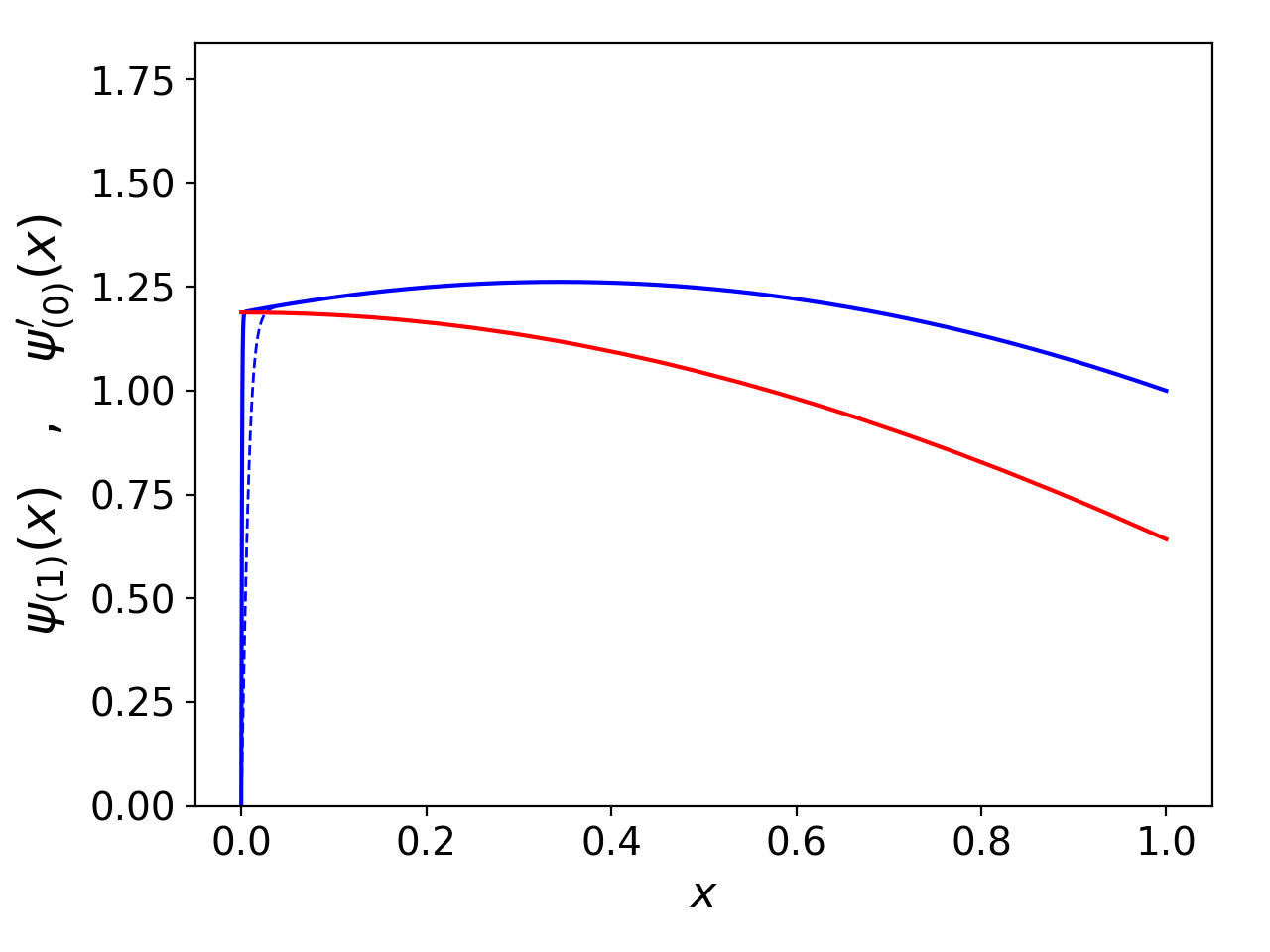}
 \includegraphics[scale=0.35]{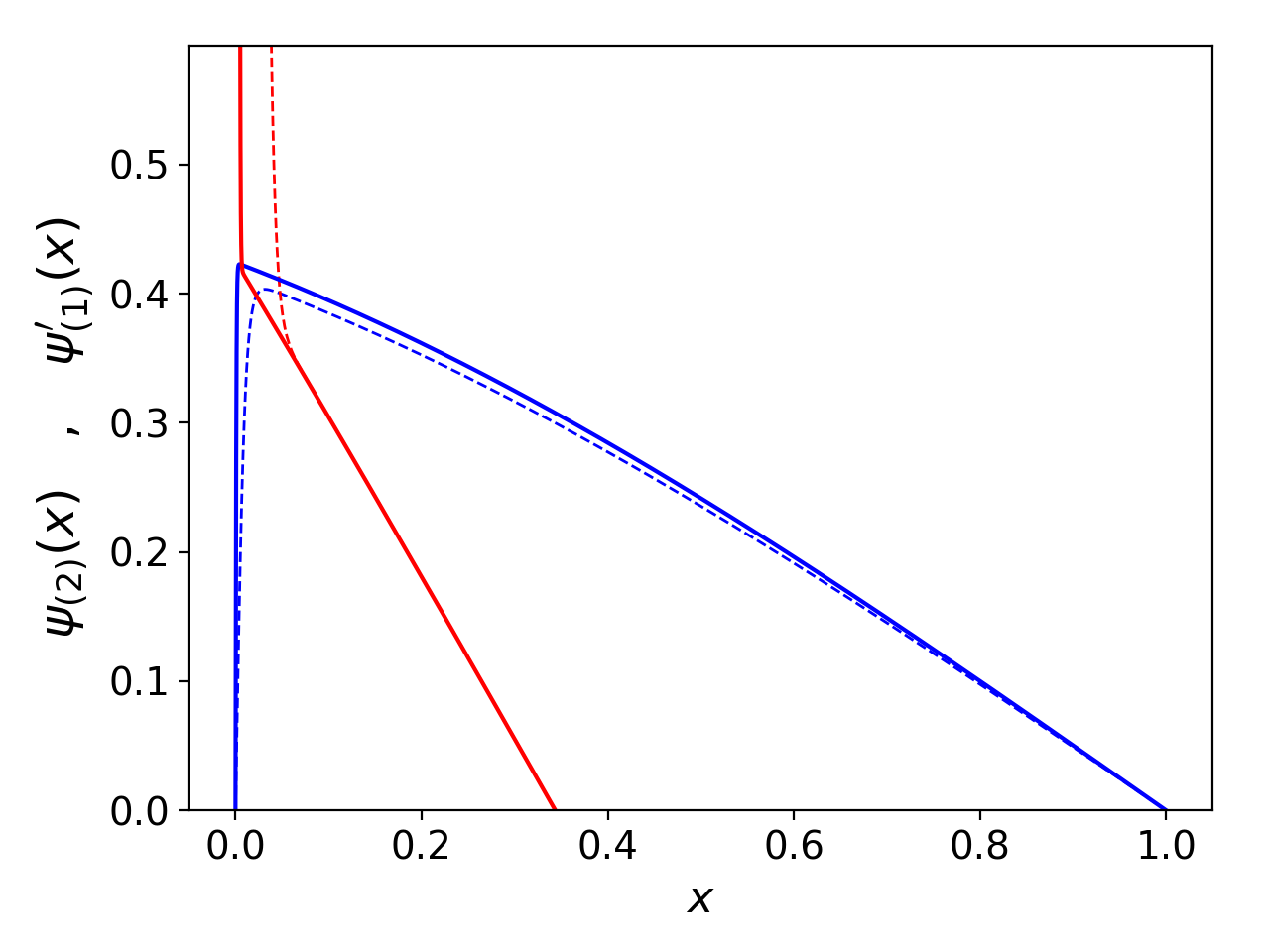}\\
 \includegraphics[scale=0.35]{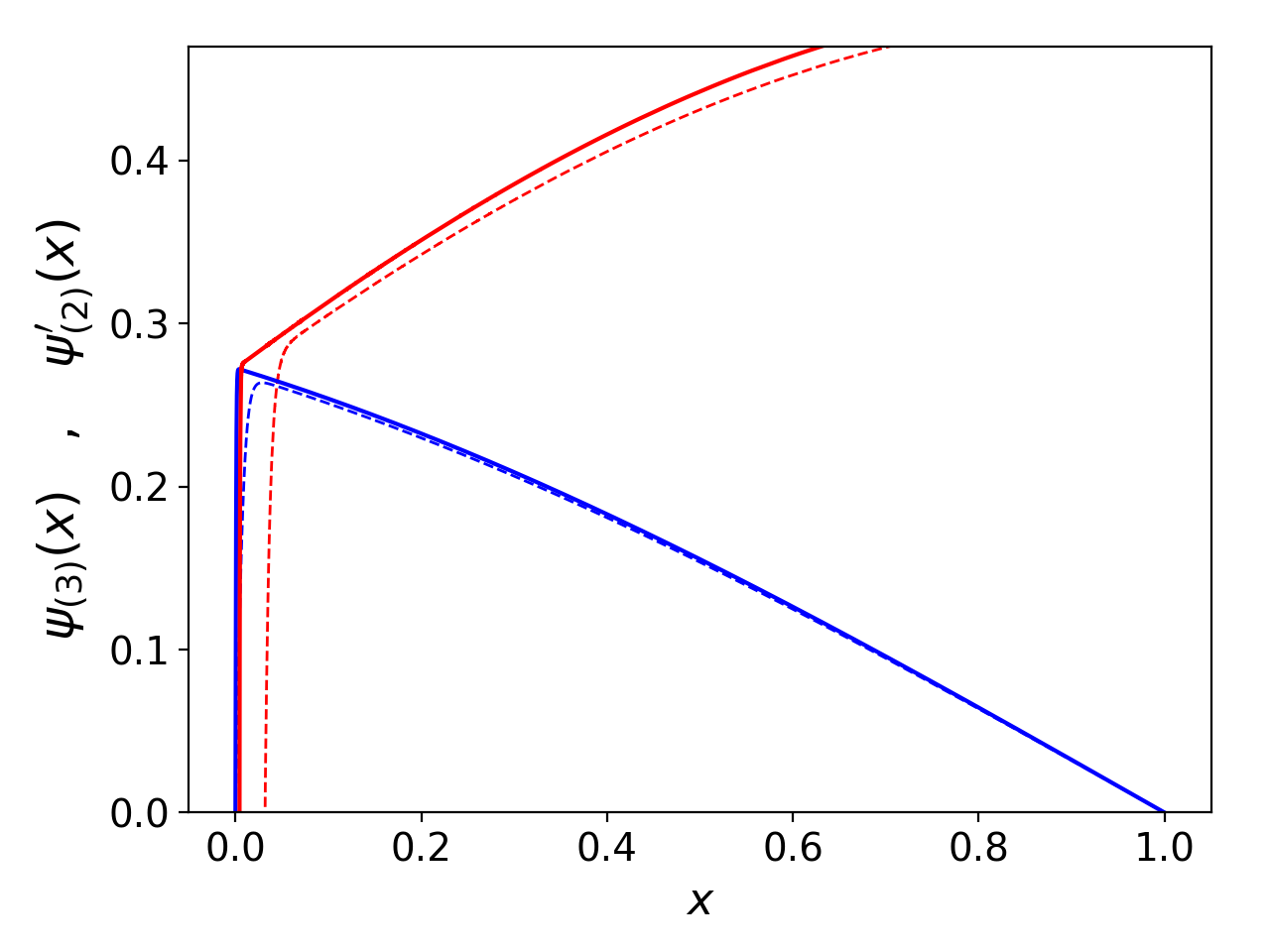}
 \includegraphics[scale=0.35]{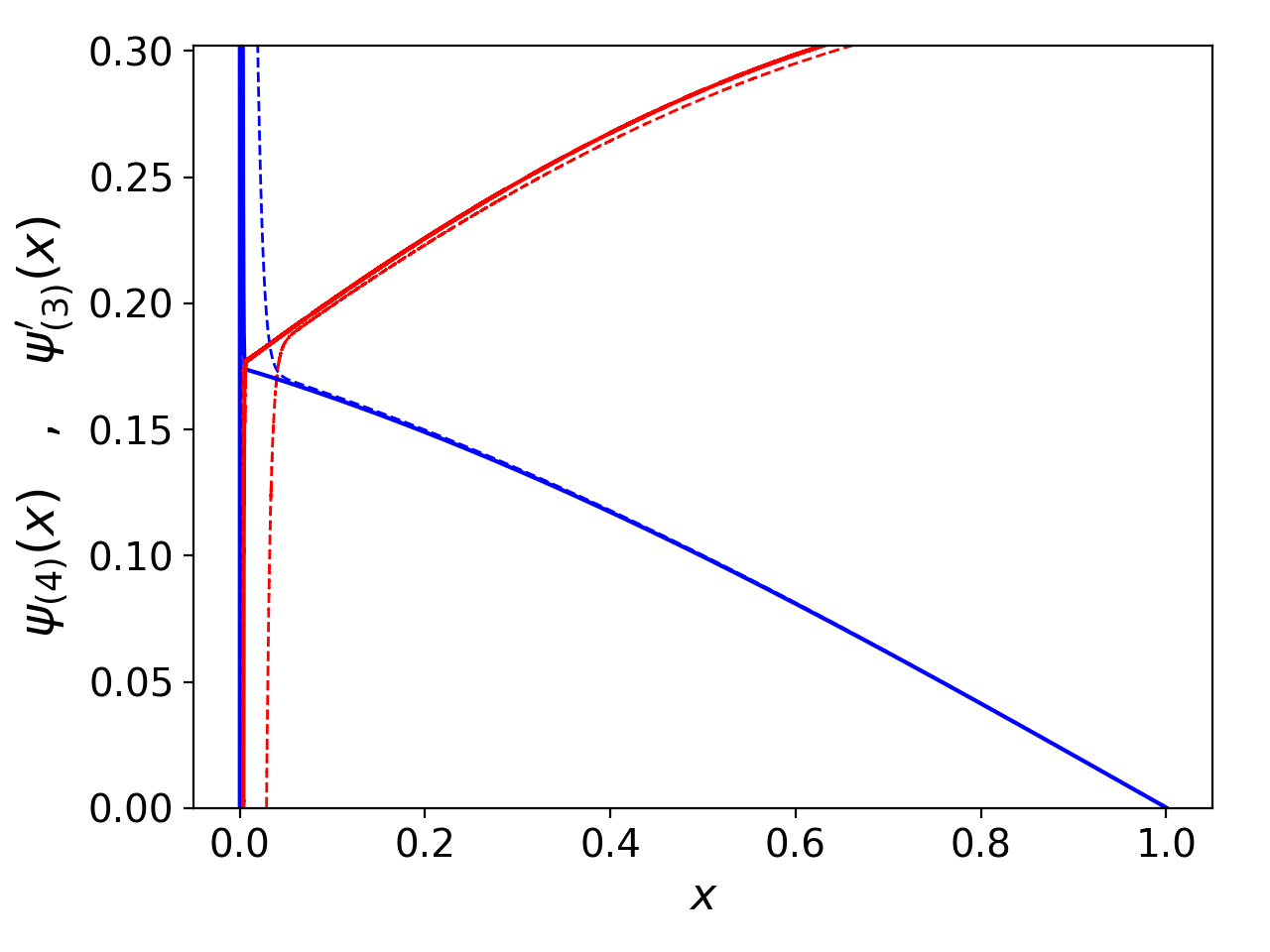}
\caption{$\psi_{(n+1)}(x)$ (blue) and $\psi_{(n)}'(x)$ (red) as a function of $x$, for $n=0,1,2,3$ in reading direction. We considered $\sigma_a(x)=\tanh(x/a)$ with $a=10^{-2}$ (dashed) and $a=10^{-3}$ (plain), $x_0=1$, and $2.10^7$ discretization points. None of the functions plotted are actually divergent at $0$. Conjecture \ref{conjecture} is satisfied if the blue and red curves take the same value for $x\to 0^+$, when $a\to 0$.}    
\label {ctdensity4}
\end{center}
\end {figure}

\section{Example of application: the Cheon-Shigehara gas \label{examp}}
In this Section we show that one can use our potential to compute the energy levels of the Lieb-Liniger model at large coupling.

\subsection{The model}
We consider a gas of $N$ fermions in interaction with the potential $V_{a,\beta}(x)$
\begin{equation}\label{hams}
H^f_{a,\beta}=-\sum_{j=1}^N \partial_{x_j}^2+2\sum_{j<k}V_{a,\beta}(x_j-x_k)\,.
\end{equation}
When $a\to 0$, this Hamiltonian becomes the Cheon-Shigehara gas \cite{CS99,GBE21}. It is equivalent to the Lieb-Liniger gas for $N$ bosons
\begin{equation}
H^b_c=-\sum_{j=1}^N \partial_{x_j}^2+2c\sum_{j<k}\delta(x_j-x_k)\,.
\end{equation}
We are going to perform a perturbative expansion of the energy levels of $H^f_{a,\beta}$, and show that when $a\to 0$ we recover the energy levels of $H^b_c$ with $c=1/\beta$.

%\subsection{Fermionic formulation of the Lieb-Liniger model}
%The Lieb-Liniger model is a non relativistic field theory in 1D whose Hamiltonian is
%\begin{equation}\label{ham}
%\begin{aligned}
%\mathcal{H}^{\rm b}_c=\int_0^L\Big[-\phi^\dagger(x)&
%\frac{\D{}^2}{\D{x}^2}\phi(x)+c\phi^\dagger(x)\phi^\dagger(x)\phi(x)\phi(x)
%\Big] \D{x}\,,
%\end{aligned}
%\end{equation}
%where $L$ is the size of the system, $c$ a coupling constant, and $\phi(x)$ a canonical Bose field satisfying equal-time
%commutation relations 
%\begin{equation}
%[\phi(x),\phi^\dagger(y)]=\delta(x-y)\,.
%\end{equation}
%In first quantization with $N$ particles the Hamiltonian of the Lieb-Liniger model is
%\begin{equation}
%H^{\rm b}_c=-\sum_{j=1}^N \partial_{x_j}^2+2c \sum_{j<k}\delta(x_j-x_k)\,.
%\end{equation}
%Let us introduce now the following Hamiltonian describing fermions interacting via potential \eqref{potential}
%\begin{equation}
%H^{\rm f}_{a,\beta}=-\sum_{j=1}^N \partial_{x_j}^2+ 2\sum_{j<k}V_{a,\beta}(x_j-x_k)\,.
%\end{equation}
%Because of the discontinuity introduce in the wave function by the potential $V_{a,\beta}$ in the limit $a\to 0$ of Theorem \ref{theorem}, it is known \cite{CS98,G60} that for $a\to 0$ this Hamiltonian is equivalent to the Lieb-Liniger model with $\beta=\frac{2}{c}$
%\begin{equation}
%H^{\rm f}_{0,2/c} \simeq H^{\rm b}_c\,,
%\end{equation}
%in the sense that the eigenfunctions $\psi^{\rm f}_{0,2/c} $ and $\psi^{\rm b}_c$ of these two Hamiltonians are related by \cite{G60}
%\begin{equation}
%\psi^{\rm f}_{0,2/c}(x_1,...,x_N)=\psi^{\rm b}_c(x_1,...,x_N)\prod_{j<k}\sign(x_j-x_k)\,.
%\end{equation}
In order to be able to apply Schr\"odinger perturbation theory, we put $H_{a,\beta}^f$ on a lattice. We consider thus a lattice model with $M=L/\kappa$ sites with periodic boundary conditions and Hamiltonian
\begin{equation}
\begin{aligned}
H&=H_0+\mathcal{V}\\
H_0&=-\frac{1}{L\kappa^2}\sum_{n=-M/2}^{M/2-1} c^\dagger_{n+1}c_{n}+c^\dagger_{n}c_{n+1}-2c^\dagger_nc_n\\
\mathcal{V}&=\frac{1}{L}\sum_{n,m=-M/2}^{M/2-1} V_{a,\beta}(n\kappa)c_{n+m}^\dagger c_m^\dagger c_{m}c_{n+m}\,,
\end{aligned}
\end{equation}
with $\{c_n,c_m^\dagger\}=\delta_{n,m}$ fermions satisfying canonical anticommutation relations. The potential $V_{a,\beta}$ is given by \eqref{potential} and $\kappa$ is a rescaling parameter. In the limit $\kappa\to 0$, this model is the second quantization formulation of the Hamilonian \eqref{hams} with the fermionic field $c_n=\sqrt{\kappa}\psi(n\kappa)$.\\

\subsection{Perturbation theory of energy levels}
 Defining the Fourier transform
\begin{equation}
c(k)=\frac{1}{\sqrt{M}}\sum_n c_n e^{ikn}\,,
\end{equation}
the free part $H_0$ is diagonalized by
\begin{equation}
H_0=\frac{4 }{L\kappa^2}\sum_{k\in {\rm R}} \sin^2(k/2)c^\dagger(k) c(k)\,,
\end{equation}
with ${\rm R}=\{\tfrac{2\pi n}{M},n=-M/2,...,M/2-1\}$.\\

Let us now fix an eigenstate of the free Hamiltonian at $\beta=0$, i.e. a set $\pmb{\lambda}\subset {\rm R}$. We will assume for simplicity it has zero momentum. 
%To compute the energy of this state perturbatively in $\beta$, we need the form factors
%\begin{equation}
%\langle\pmb{\mu}| c_{n+m}^\dagger c_m^\dagger c_{n+m}c_m|\pmb{\lambda}\rangle=\begin{cases}
%\frac{1}{M}\sum_{\lambda\neq\mu\in\pmb{\lambda}}e^{in(\lambda-\mu)}-1\qquad \text{if }\pmb{\lambda}=\pmb{\mu}\\
%\frac{1}{M}\sum_{\lambda\neq\mu\in\pmb{\lambda}}e^{in(\lambda-\mu)}-1\qquad \text{if }\pmb{\mu}\perp\pmb{\lambda}=\{h_1,h_2,h_1+p,h_2-p\}\\
%0\qquad \text{otherwise}
%\end{cases}
%\end{equation}

%, described by a root density $\frac{1}{\kappa}\tilde{\rho}(\lambda)$ and a hole density $\frac{1}{\kappa}\tilde{\rho}_h(\lambda)=\tfrac{1}{2\pi\kappa}-\frac{1}{\kappa}\tilde{\rho}(\lambda)$ in the thermodynamic limit $L\to\infty$ at fixed $\kappa$. 
Using perturbation theory, we find the energy density at order $\beta^2$
\begin{equation}
E(\beta)=E_0+E_1+E_2+\mathcal{O}(\beta^3)\,,
\end{equation}
with
\begin{equation}
\begin{aligned}
E_0&=\frac{4}{\kappa^2L}\sum_{\lambda\in\pmb{\lambda}}\sin^2(\lambda/2)\\
E_1&=\frac{1}{L^2}\sum_{\lambda,\mu\in\pmb{\lambda}} \left[\tilde{V}_{a,\beta}(0)-\tilde{V}_{a,\beta}(\lambda-\mu)\right]\\
E_2&=\frac{\kappa^2}{4L^3}\sum_{\substack{\lambda,\mu\in\pmb{\lambda}\\ \nu\in{\rm R}\\ \lambda+\nu\notin\pmb{\lambda}\\\mu-\nu\notin\pmb{\lambda}}}\frac{\left[\tilde{V}_{a,\beta}(\lambda-\mu+\nu)-\tilde{V}_{a,\beta}(\nu)\right]^2}{\sin^2 \tfrac{\lambda}{2}+\sin^2 \tfrac{\mu}{2}-\sin^2 \tfrac{\lambda+\nu}{2}-\sin^2 \tfrac{\mu-\nu}{2}}\,,
%E_0&=\frac{4}{\kappa^2L}\int_{-\pi}^\pi \sin^2(\lambda/2)\tilde{\rho}(\lambda)\D{\lambda}\\
%E_1&=-\frac{1}{2\kappa^2}\int_{-\pi}^\pi \int_{-\pi}^\pi \left[\tilde{V}_{a,\beta}(\lambda-\mu)-\tilde{V}_{a,\beta}(0)\right]\tilde{\rho}(\lambda)\tilde{\rho}(\mu)\D{\lambda}\D{\mu}\\
%E_2&=-\frac{1}{4\kappa^2}\int_{-\pi}^\pi \int_{-\pi}^\pi \int_{-\pi}^\pi \frac{\left[\tilde{V}_{a,\beta}(\lambda-\mu+\nu)-\tilde{V}_{a,\beta}(\nu)\right]^2}{\nu(\lambda-\mu+\nu)}\tilde{\rho}(\lambda)\tilde{\rho}(\mu)2\pi\tilde{\rho}_h(\lambda+\nu)\tilde{\rho}_h(\mu-\nu)\D{\lambda}\D{\mu}\D{\nu}\,,
\end{aligned}
\end{equation}
with 
\begin{equation}
\tilde{V}_{a,\beta}(\lambda)=\kappa\sum_{n=-M/2}^{M/2-1} V_{a,\beta}( n\kappa)e^{i\lambda n}\,.
\end{equation}

\subsection{Low-momentum and thermodynamic limit}
We now rescale the momenta as $\lambda\to \lambda\kappa$, and take both the thermodynamic limit $L\to\infty$ and the low-momentum limit $\kappa\to 0$. We  denote $\rho(\lambda)$ the particle density and $\rho_h(\lambda)=\tfrac{1}{2\pi}-\rho(\lambda)$ the hole density. We obtain
\begin{equation}
\begin{aligned}
E_0&=\int_{-\infty}^\infty \lambda^2\rho(\lambda)\D{\lambda}\\
E_1&=\iint_{-\infty}^\infty \left[\hat{V}_{a,\beta}(0)-\hat{V}_{a,\beta}(\lambda-\mu)\right]\rho(\lambda)\rho(\mu)\D{\lambda}\D{\mu}\\
E_2&=-\frac{1}{2}\iiint_{-\infty}^\infty \frac{\left[\hat{V}_{a,\beta}(\lambda-\mu+\nu)-\hat{V}_{a,\beta}(\nu)\right]^2}{\nu(\lambda-\mu+\nu)}\rho(\lambda)\rho(\mu)2\pi\rho_h(\lambda+\nu)\rho_h(\mu-\nu)\D{\lambda}\D{\mu}\D{\nu}\,,
\end{aligned}
\end{equation}
with the Fourier transform
\begin{equation}
\hat{V}_{a,\beta}(\lambda)=\int_{-\infty}^\infty V_{a,\beta}(x)e^{i\lambda x}\D{x}\,.
\end{equation}
We will denote $\mathcal{E}_a^{(n)}$ the truncation at order $\beta^n$ of the series of $E(\beta)$ in $\beta$ at fixed $a$.

\subsection{Regularization limit}
Let us now study the regularization limit $a\to 0$ of the expansion in $\beta$. We expand at fixed $a$
\begin{equation}
V_{a,\beta}(x)=\beta \frac{\sigma''(x/a)}{xa^2}-\beta^2\frac{\sigma''(x/a)\sigma(x/a)}{x^2a^2}+\mathcal{O}(\beta^3)\,,
\end{equation}
from which we find
\begin{equation}
\hat{V}_{a,\beta}(k)-\hat{V}_{a,\beta}(0)=\beta k^2-\frac{\beta^2 k^2}{4a\pi}\int_{-\infty}^\infty [\widehat{\sigma'}(\omega)]^2\D{\omega}+\mathcal{O}(a)+\mathcal{O}(\beta^3)\,.
\end{equation}
Here the order of the expansions is first $\beta\to 0$ at fixed $a$ and then $a\to 0$. Hence it yields
\begin{equation}
E_1=-\beta \int_{-\infty}^\infty \int_{-\infty}^\infty(\lambda-\mu)^2\rho(\lambda)\rho(\mu)\D{\lambda}\D{\mu}\left[1-\frac{\beta}{4a\pi}\int_{-\infty}^\infty [\widehat{\sigma'}(\omega)]^2\D{\omega} \right]+\mathcal{O}(a)+\mathcal{O}(\beta^3)\,.
\end{equation}
It is divergent in $1/a$ at order $\beta^2$. To investigate $E_2$, we write
\begin{equation}
\hat{V}_{a,\beta}(x)-\hat{V}_{a,\beta}(0)=\frac{\beta}{a^2}F(xa)+\mathcal{O}(\beta^2)\,,
\end{equation}
with
\begin{equation}
F(x)=\int_0^{x} z \widehat{\sigma'}(z)\D{z}\,.
\end{equation}
We plug this in the expression for $E_2$, separate $\rho_h(\lambda)=\tfrac{1}{2\pi}-\rho(\lambda)$ and perform a change of variable $\nu\to \nu/a$ for the case where $\tfrac{1}{2\pi}$ is taken twice. It yields
\begin{equation}
\begin{aligned}
E_2=E_{2}^{\rm reg}+E_{2}^{\rm sing}
\end{aligned}
\end{equation}
where
\begin{equation}
E_{2}^{\rm sing}=-\frac{\beta^2}{4a^3\pi} \int_{-\infty}^\infty \int_{-\infty}^\infty\rho(\lambda)\rho(\mu)\frac{(F(\nu+a(\lambda-\mu))-F(\nu))^2}{\nu(\nu+a(\lambda-\mu))}\D{\lambda}\D{\mu}\D{\nu}\,,
\end{equation}
and
\begin{equation}\label{reg}
\begin{aligned}
E_{2}^{\rm reg}=\frac{\beta^2}{2}&\int_{-\infty}^\infty \int_{-\infty}^\infty \int_{-\infty}^\infty  \frac{(\lambda-\mu)^2(\lambda-\mu+2\nu)^2}{\nu(\lambda-\mu+\nu)}\rho(\lambda)\rho(\mu)\\
&\times\Big[\rho(\lambda+\nu)+\rho(\mu-\nu)-2\pi \rho(\lambda+\nu)\rho(\mu-\nu)\Big]\D{\lambda}\D{\mu}\D{\nu}+\mathcal{O}(a)+\mathcal{O}(\beta^3)\,.
\end{aligned}
\end{equation}
Focusing on $E_{2}^{\rm sing}$ first, we write
\begin{equation}
\frac{1}{\nu(\nu+a(\lambda-\mu))}=\frac{1}{a(\lambda-\mu)}\left( \frac{1}{\nu}-\frac{1}{\nu+a(\lambda-\mu)}\right)\,,
\end{equation}
and do a change of variable $\nu\to \nu-a(\lambda-\mu)$ for the second fraction, which yields, expanding $F$ at $\nu$
\begin{equation}
E_2^{\rm sing}=-\frac{\beta^2}{4a\pi} \int_{-\infty}^\infty \int_{-\infty}^\infty(\lambda-\mu)^2\rho(\lambda)\rho(\mu)\D{\lambda}\D{\mu}\left[\int_{-\infty}^\infty [\widehat{\sigma'}(\omega)]^2\D{\omega} \right]+\mathcal{O}(a)+\mathcal{O}(\beta^3)\,.
\end{equation}
Remarkably, the form of the potential \eqref{potential} exactly makes the divergent piece in $1/a$ compensate between $E_1$ and $E_2$. Let us now focus on the regular part \eqref{reg}. Doing a change of variable $\lambda'=\lambda+\nu$, $\mu'=\mu-\nu$ and $\nu'=-\nu$ one sees that the term with four $\rho$'s vanishes. Doing the change of variable $\mu'=\lambda,\lambda'=\mu,\nu'=-\nu$, one sees that the terms with three $\rho$'s are equal, so that
\begin{equation}
\begin{aligned}
E_{2}^{\rm reg}&=\beta^2\int_{-\infty}^\infty \int_{-\infty}^\infty \int_{-\infty}^\infty  \frac{(\lambda-\mu)^2(2\nu-\lambda-\mu)^2}{(\nu-\lambda)(\nu-\mu)}\rho(\lambda)\rho(\mu)\rho(\nu)\D{\lambda}\D{\mu}\D{\nu}+\mathcal{O}(a)+\mathcal{O}(\beta^3)\\
%&=4\beta^2\int_{-\infty}^\infty \int_{-\infty}^\infty \int_{-\infty}^\infty  (\lambda-\mu)^2\rho(\lambda)\rho(\mu)\rho(\nu)\D{\lambda}\D{\mu}\D{\nu}\\
%&\quad-2\beta^2\int_{-\infty}^\infty \int_{-\infty}^\infty \int_{-\infty}^\infty  \frac{(\lambda-\mu)^3}{(\nu-\mu)}\rho(\lambda)\rho(\mu)\rho(\nu)\D{\lambda}\D{\mu}\D{\nu}+\mathcal{O}(a)+\mathcal{O}(\beta^3)\\
&=\frac{3}{2}\beta^2\int_{-\infty}^\infty \int_{-\infty}^\infty(\lambda-\mu)^2\rho(\lambda)\rho(\mu)\rho(\nu)\D{\lambda}\D{\mu}\D{\nu}+\mathcal{O}(a)+\mathcal{O}(\beta^3)\,,
\end{aligned}
\end{equation}
where we used that $\rho$ is even, after several manipulations of the type
\begin{equation}
\int_{-\infty}^\infty \int_{-\infty}^\infty\frac{\mu}{\mu-\nu}\rho(\mu)\rho(\nu)\D{\mu}\D{\nu}=\frac{1}{2}\int_{-\infty}^\infty \int_{-\infty}^\infty \rho(\mu)\rho(\nu)\D{\mu}\D{\nu}\,,
\end{equation}
to eliminate all the fractions. Gathering everything, we obtain that $\mathcal{E}^{(2)}_a$ has a well-defined limit when $a\to 0$ that reads
\begin{equation}\label{res1c}
\mathcal{E}^{(2)}_0=\int_{-\infty}^\infty \lambda^2\rho(\lambda)\D{\lambda}\left(1-2\beta \mathcal{D}+3\beta^2\mathcal{D}^2 \right)\,,
\end{equation}
with
\begin{equation}
\mathcal{D}=\int_{-\infty}^\infty \rho(\lambda)\D{\lambda}\,.
\end{equation}
This is indeed the Bethe ansatz result for the Lieb-Liniger model with $\beta=\frac{2}{c}$ at order $1/c^2$.

\section{Summary and discussion}
We have shown that the potential \eqref{potential} produces in the limit $a\to 0$ a discontinuity in the fermionic wave function given by \eqref{shrod2}. This provides thus a regularization of the only pointlike interaction between fermions given by a self-adjoint extension in \eqref{gen}. The particularity of this regularization is that contrary to all previous ones, the potential is small when the discontinuity is small, implementing the strong-weak duality between fermions and bosons in 1D given by the Girardeau mapping. This opens the way to perturbative studies of strongly coupled bosons, as done in the companion paper \cite{GBE21}. Besides, we gave in Section \ref{pertee} numerical evidence for the possibility of commuting the limit $a\to 0$ and the expansion at small $\beta$ in the two-particle wave function. In Section \ref{examp}, we showed that in the many-body case the two limits commute at least up to order $\beta^2$ included.

The attractive case $\beta<0$ deserves some attention, as the potential given in this paper is not applicable to this case. In fact, one can see that imposing to preserve the strong-weak duality and being able to commute the limit $a\to 0$ and the expansion in $\beta$ essentially constrains the potential to be of the form \eqref{potential}, with a possible $\beta$ dependence in $\sigma$. But then one sees that in the case $\beta<0$ there will always be a singularity in the potential for $x\approx |\beta|$. Cutting the potential for $x>\beta$ to avoid this singularity would spoil the commutation of the two limits. This suggests that in the case $\beta<0$ a regularization consistent with the strong-weak duality might not exist. This important difference would reflect the very different physics of the Lieb-Liniger model depending on the sign of $\beta$, with bound states appearing only for $\beta<0$. But we leave a precise study of these aspects for future work.

\appendix
\section{Appendix}
\subsection{Study of \eqref{eqdelta3}}\label{app}
In this Appendix we study the equation
\begin{equation}\label{eqdelta5}
\psi_a''(x)+k^2\psi_a(x)=\beta \sigma_a''(x)\psi_a'(x)+\beta\sigma_a'(x)\psi_a''(x)\,,
\end{equation}
with $\sigma_a(x)\to \sign(x)$ when $a\to 0$.

\subsubsection{Absence of discontinuity \label{absenceed}}
Integrating \eqref{eqdelta5} between $x_0$ and $x$, one has
\begin{equation}\label{newo}
\psi'_a(x)-\beta \sigma'_a(x)\psi_a'(x)=\psi'_a(x_0)-\beta \sigma'_a(x_0)\psi_a'(x_0)-k^2\int_{x_0}^x \psi_a(y)\D{y}\,.
\end{equation}
 Let us now specify
 \begin{equation}
\sigma'_a(x)=\frac{1}{2a}\1_{|x|<a}\,.
\end{equation}
Integrating between $-a$ and $a$ one finds
\begin{equation}
(\psi_a(a)-\psi_a(-a))(1-\tfrac{\beta}{2a})=\mathcal{O}(a)\,,
\end{equation}
which imposes that
\begin{equation}\label{means}
\underset{a\to 0}{\lim}( \psi_a(a)-\psi_a(-a))=0\,.
\end{equation}
Now, integrating \eqref{newo} between $a$ and $a+\epsilon$ for $\epsilon>0$ we find
\begin{equation}\label{means2}
\psi_a(a+\epsilon)-\psi_a(a)=\mathcal{O}(\epsilon,a)\,.
\end{equation}
Defining
\begin{equation}
\psi(x)=\underset{a\to 0}{\lim}\,\psi_a(x)\,,
\end{equation}
we find that \eqref{means2} implies
\begin{equation}
\underset{a\to 0}{\lim}\,\psi_a(a)=\psi(0^+)\,.
\end{equation}
Similarly one has
\begin{equation}
\underset{a\to 0}{\lim}\,\psi_a(-a)=\psi(0^-)\,.
\end{equation}
Hence \eqref{means} means
\begin{equation}
\psi(0^+)=\psi(0^-)\,,
\end{equation}
and the function is continuous at zero.

\subsubsection{A discontinuity present perturbatively but absent non-perturbatively}
At first order in $\beta$ at fixed $a$, the odd solution to \eqref{eqdelta5} is
\begin{equation}
\begin{aligned}
&e^{ikx}-e^{-ikx}\\
&+\frac{\beta}{2ik}\int_0^x(e^{ik(x-y)}-e^{-ik(x-y)})(ik\sigma_a''(y)(e^{iky}+e^{-iky})-k^2\sigma_a'(y)(e^{iky}-e^{-iky}))\D{y}\,,
\end{aligned}
\end{equation}
whose limit $a\to 0$ is
\begin{equation}
e^{ikx}-e^{-ikx}+ik\beta \cos(kx)\sign(x)\,.
\end{equation}
This function has a discontinuity at $x=0$ that satisfies \eqref{shrod5} indeed. This is not in contradiction with Appendix \ref{absenceed}: a discontinuity can be present perturbatively, but absent non-perturbatively. \\

In order to illustrate this phenomenon, we consider the following simple tractable example
\begin{equation}
f_a''-\beta(\sigma_a' f_a')'=0\,.
\end{equation}
At leading order in $\beta$ at fixed $a$ one finds
\begin{equation}\label{pert}
f_a(x)=A(x+\beta\sigma_a(x))+B+\mathcal{O}(\beta^2)\,,
\end{equation}
with $A,B$ integration constants, which has indeed a discontinuity at zero in the limit $a\to 0$. However, the non-pertubative solution is
\begin{equation}
f_a(x)=A\int_{-1}^x \frac{\D{y}}{1-\beta\sigma_a'(y)}+B\,,
\end{equation}
with a principal value in case of a vanishing denominator. Let us investigate the behaviour at zero of this function in the limit $a\to 0$ if we specify for example
\begin{equation}
\sigma_a'(x)=\frac{2}{\pi}\frac{a}{a^2+x^2}\,.
\end{equation}
Then for a $\epsilon>0$ fixed
\begin{equation}
\int_{-1}^x \frac{\D{y}}{1-\beta\sigma_a'(y)}=\int_{-1}^{-\epsilon} \frac{\D{y}}{1-\beta\sigma_a'(y)}+\int_{-\epsilon}^\epsilon \frac{\D{y}}{1-\beta\sigma_a'(y)}+\int_{\epsilon}^x \frac{\D{y}}{1-\beta\sigma_a'(y)}\,.
\end{equation}
At fixed $\epsilon>0$ and $x>\epsilon$, the first and last integral go to 
\begin{equation}
\int_{-1}^{-\epsilon} \D{y}+\int_{\epsilon}^x \D{y}\,,
\end{equation}
while the second integral
\begin{equation}
\int_{-\epsilon}^\epsilon \frac{\D{y}}{1-\beta\sigma_a'(y)}=2\epsilon-\frac{4a\beta}{\pi \sqrt{\tfrac{2a\beta}{\pi }-a^2}}\Re \text{argtanh } \frac{\epsilon}{\sqrt{\tfrac{2a\beta}{\pi }-a^2}}
\end{equation}
goes to $2\epsilon$ in the limit $a\to 0$. The behaviour for $x<-\epsilon$ is immediate, so that the full integral converges to $\int_{-1}^x \D{y}$. Hence in the limit $a\to 0$
\begin{equation}
f_0(x)=Ax+A+B\,,
\end{equation}
which has no discontinuity.

%\subsection{Presence of the cusp in \eqref{shrod1} perturbatively in $\gamma$}\label{app2}
%\tableofcontents

%
%\vspace{10pt}
%\noindent\rule{\textwidth}{1pt}
%\tableofcontents\thispagestyle{fancy}
%\noindent\rule{\textwidth}{1pt}
%\vspace{10pt}

\renewcommand\Affilfont{\fontsize{9}{10.8}\itshape}

\end{document}